\documentclass[article]{imsart}

\usepackage{amsthm,amsmath}
\usepackage{amsfonts}
\RequirePackage[numbers]{natbib}
\RequirePackage[colorlinks,citecolor=blue,urlcolor=blue]{hyperref}

\startlocaldefs

\theoremstyle{plain}

\newtheorem{theorem}{theorem}

\newtheorem{algorithm}[theorem]{Algorithm}

\newtheorem{corollary}[theorem]{Corollary}

\newtheorem{lemma}[theorem]{Lemma}

\newtheorem{remark}[theorem]{Remark}

\newcommand{\ehb}{\hat{\eta}^N }
\newcommand{\fb}{f }

\endlocaldefs

\begin{document}

\begin{frontmatter}

\title{Stability of conditional Sequential Monte Carlo}
\thankstext{t1}{Bernd Kuhlenschmidt's work is supported by the UK Engineering and Physical Sciences Research Council (EPSRC) grant EP/H023348/1 for the University of Cambridge Centre for Doctoral Training, the Cambridge Centre for Analysis. This work was completed during the Ph.D.\ candidature of the first author.  Title of dissertation: \emph{On the Stability of Sequential Monte Carlo Methods for Parameter Estimation.} Examiners: Prof.\ Arnaud Doucet and Prof.\ S.J.\ Godsill. Examination date:  29 January 2015.}
\runtitle{Stability of conditional Sequential Monte Carlo}

\author{\fnms{Bernd} \snm{Kuhlenschmidt}\thanksref{t1}\corref{}\ead[label=e1]{bkuhlensch@gmail.com}}
\address{Cambridge Centre for Analysis\\ University of Cambridge (UK)\\ \printead{e1}}

\author{\fnms{Sumeetpal S.} \snm{Singh}\ead[label=e2]{sss40@cam.ac.uk}}
\address{Department of Engineering\\ University of Cambridge (UK)\\  \printead{e2}}

\runauthor{B. Kuhlenschmidt and S. Singh}

\begin{abstract}
The particle Gibbs (PG) sampler is a Markov Chain Monte Carlo (MCMC) algorithm, which uses an interacting particle system to perform the Gibbs steps. Each Gibbs step consists of simulating a particle system conditioned on one particle path. It relies on a conditional Sequential Monte Carlo (cSMC) method to create the particle system. We propose a novel interpretation of the cSMC algorithm as a perturbed Sequential Monte Carlo (SMC) method and apply  telescopic decompositions developed for the analysis of SMC algorithms \cite{delmoral2004} to derive a bound for the distance between the expected sampled path from cSMC and the target distribution of the MCMC algorithm. This can be used to get a uniform ergodicity result. In particular, we can show that the mixing rate of cSMC can be kept constant by increasing the number of particles linearly with the number of observations. Based on our decomposition, we also prove a central limit theorem for the cSMC Algorithm, which cannot be done using the approaches in \cite{Andrieu2013} and \cite{Lindsten2014}.  
\end{abstract}





\end{frontmatter}

\section{Introduction}
Markov Chain Monte Carlo (MCMC) methods are a powerful and widely used tool in computational statistics, in particular for the Bayesian estimation of static  parameters in Hidden Markov Models (HMM). There has been a surge of interest in applying these methods to increasingly sophisticated and high dimensional models. Standard MCMC methods might fail in these scenarios since they might require the computation of intractable quantities. Recently \cite{Andrieu2010} have introduced particle Markov Chain Monte Carlo (pMCMC) methods, which are based on the idea of using a SMC method inside the MCMC. This combination of SMC and MCMC proves to be useful in several applications, since the straightforward implementation of MCMC methods such as Gibbs samplers or Metropolis Hastings might lead to very slowly mixing algorithms, a problem that pMCMC seems to solve to some extent. 

PMCMC has been applied to a wide range of areas such as Hydrology \cite{Vrugt2013}, Finance \cite{Pitt2011}, Systems Biology
\cite{Golightly2011}, Social networks \cite{Everitt2012}, and Electricity forecasting \cite{Launay2013}. 

The literature currently focusses on two types of pMCMC: particle Metropolis Hastings (pMH) and particle Gibbs (pG). In pMH, an SMC method is used to compute the likelihood for the proposed parameter values in the acceptance step. pMH uses an SMC method
to obtain an unbiased estimator of the likelihood terms needed to compute the acceptance probability for a new proposal. 

This method can be regarded as a perturbed Metropolis Hastings Algorithm and has been analysed in \cite{Andrieu2013}, \cite{Andrieu2010a} and \cite{Andrieu2009}. 
PG, on the other hand, is based on conditioning an SMC sampler on a reference trajectory and using the weights of the simulated particles at the final time step to select one path. This path constitutes a new sample of the pG Algorithm. This Algorithm can be viewed as a Gibbs sampler for a model where the random variables generated by the SMC sampler are treated as auxiliary variables. Another view would be that the pG Algorithm defines a Markov kernel with correct invariant measure, i.e. the invariant measure is the Bayesian posterior from which samples are needed. This is the view further developed in this paper.

One of the main issues with using SMC to obtain the smoothing distribution of $X_{1:n}$ conditional on the observations $Y_{1:n}$ is the path degeneracy phenomenon which means that if the SMC Algorithm is applied for a large $n$ eventually all particles will have a common ancestor. This problem also plagues pG. One solution to this is to resample in the backward direction as suggested in \cite{Whiteley2010}. \cite{Chopin2012} show that this leads to a reversible Markov kernel. 

Recently, many contributions have been made to improving and tuning pMCMC, e.g. an analysis of different resampling schemes \cite{Chopin2012}, the development of second order pMCMC methods \cite{citeulike:12797899}, backward resampling \cite{Link2012} and \cite{Whiteley2010} and parallelization \cite{Henriksen_parallelimplementation}.

The accuracy of the SMC Algorithm can be improved by increasing the number of particles $N$, which also leads to an increase in computational effort. In the context of particle MCMC one can increase $N$ to achieve better mixing properties of the Algorithm or increase the number of MCMC steps. In order to evaluate how to choose $N$, performance bounds and stability results are warranted.  \cite{Andrieu2010a} analyse particle MH as a noisy version of an MH Algorithm and obtain ergodicity results. \cite{citeulike:11421643} shows that it is optimal in some sense to increase $N$ at least linearly with the number of observations $n$. 
The question is if there are similar results for pG. \cite{Chopin2012} have estalished a stability result for the cSMC kernel using a coupling argument, showing that the difference between two differently initialized cSMC kernels goes to 0 as the number of particles goes to infinity. \cite{Andrieu2013} have derived a bound for the ergodicity of cSMC and the pGibbs kernel. They give explicit rates using strong mixing assumptions, which imply that by choosing $N$ linearly dependend on $n$ one can stabilize cSMC.  They also derive  necessary conditions for uniform ergodicity and consider a complete pG scheme for parameter estimation, sampling both paths and static variables.
\cite{Lindsten2014} derive a similar ergodicity bound under strong mixing assumptions. They also derive an ergodicity result for more general mixing assumptions. The conclusion of this result is that it is sufficient to increase the number of particles $N$ as $n^\delta$ where $\delta \geq 1$ can be computed explicitly. 
The proofs of both \cite{Andrieu2013} and \cite{Lindsten2014} are based on a minorization of the cSMC kernel, even though the details of the proof techniques are different. 

\subsection{Contributions}
In this paper we employ a different approach to \cite{Andrieu2013}, \cite{Lindsten2014}, \cite{Chopin2012}   to obtain stability and ergodicity results. We directly consider the difference between the cSMC kernel and the target measure and apply a telescopic decomposition to obtain a bound for this. Telescopic decompositions are a standard tool in the analysis of SMC algorithms and can be applied to obtain the approximation error, central limit theorems, concentration of measure, propagation of chaos and various other results, see  \cite{delmoral2004}, \cite{delmoral:hal-00932211}. We have recently come to know about the very recent work \cite{P.DelMoralR.Kohn2014} where standard results for SMC have been provided for the cSMC, e.g. propagation of chaos. The proof techniques and results in this substantial work differ from ours.
We show that these decompositions asymptotically lead to the same ergodicity bounds as in \cite{Andrieu2013} and \cite{Lindsten2014}. We also use a telescopic decomposition to derive a central limit theorem for cSMC, which is novel. An interesting corollary of this is that the particle system created by the cSMC Algorithm exhibits the same asymptotic behavior as a conventional SMC Algorithm. 

This paper is organised as follows. Section 2 introduces the notation and describes the cSMC Algorithm which is analysed in the paper. Section 3 analyses the perturbation of the Boltzmann-Gibbs transformation which originates from conditioning the cSMC Algorithm on a path. Section 4 derives a general stability result for cSMC using a telescopic decomposition. Section 5 derives explicit error bounds for this stability result using mixing assumptions. Section 6 proves a central limit theorem. Section 7 discusses the results of this paper and compares them to existing results in the literature. 

We have recently come to know about the very recent work \cite{P.DelMoralR.Kohn2014}  where standard results for SMC have been provided for cSMC, e.g. propagation of chaos. The proof techniques and results in this substantial work differs from ours.

\section{Definitions and notation}\label{sect:def}
Let $(X_t)_{t\geq 0}$ be a discrete-time $\mathcal{X}$-valued Markov chain with initial law $m_0 (dx_0)=m_0 ( x_0 ) dx_0 $ and transition density $q( x_{k+1} \left| x_k \right   )$. Let $( G_t )_{t\geq 0}: \mathcal{X} \rightarrow \mathbb{R}$ be a sequence of potential functions.  In the context of HMM this would be $G_t (x_t) = g_t (x_t, y_t )$ where  $y_t \in \mathcal{Y}$ are the observations of the hidden state $x_t$ and $(g_t )_{t\geq 0 }$ are the observation likelihood densities. 

One can define the following path measures
\[
\eta_{n}(dx_{0:n})=\frac{m_{0}(dx_{0})\prod\limits _{i=0}^{n-1}dx_{i+1}G_{i}(x_{i})q(x_{i+1}\left|x_{i}\right)}{\int m_{0}(dx_{0})\prod\limits _{i=0}^{n-1}dx_{i+1}G_{i}(x_{i})q(x_{i+1}\left|x_{i}\right)}
\]

\[
\pi_{n}(dx_{0:n})=\frac{m_{0}(dx_{0})\prod\limits _{i=0}^{n-1}dx_{i+1}G_{i}(x_{i})q(x_{i+1}\left|x_{i}\right) G_n (x_n )}{\int m_{0}(dx_{0})\prod\limits _{i=0}^{n-1}dx_{i+1}G_{i}(x_{i})q(x_{i+1}\left|x_{i}\right) G_n (x_n )}
\]

and write $\eta_0 (dx_0) = m_0 (dx_0 )$. For HMMs $\pi_n $ is the law of $X_{0:n}$ given $Y_{0:n}=y_{0:n}$ whereas $\eta_n$ is the law of $X_{0:n}$ given $Y_{0:n-1}=y_{0:n-1}$.

In the following let $f_{n}:\mathcal{X}^{n+1} \rightarrow \mathbb{R}$  be a measurable function of
the path space with $\left\Vert f_{n}\right\Vert _{\infty}\leq1$. Let 
\begin{align*}
&Q_{k,n}  (f_n )(x_{0:k})\\
 & =\int dx_{k+1:n}G_{k}(x_{k})q(x_{k+1}\left|x_{k}\right)G_{k+1}(x_{k+1})q(x_{k+2}\left|x_{k+1}\right)\cdots G_{n-1}(x_{n-1})q(x_{n}\left|x_{n-1}\right) f_n (x_{0:n} )
\end{align*}
and
\begin{align*}
P_{k,n}  (f_{n})(x_{0:k})  = \frac{Q_{k,n} (G_n f_n ) (x_{0:k})}{   Q_{k,n} (G_n ) (x_{0:k})   }
\end{align*}

The aim behind particle Gibbs is to construct a Markov kernel for sampling from the target distribution $\pi_n$ using a cSMC Algorithm.  
The basic idea behind cSMC is to define an extended target distribution $\pi^N_T ( dx_{0:T}^{1:N} , da_{0:T-1}^{1:N} , dn^{*}  )$ defined below  using auxiliary variables  and to construct a Gibbs sampler targeting this measure. 

Let $\rho_t (z_t^{1:N}, da_t^{1:N})$ be the resampling distribution of the SMC method. For example for multinomial resampling as in Algorithm \ref{alg:csmcparticles} below this is $\rho_t (z_t^{1:N}, da_t^{1:N})= \prod\limits_{i=1}^N \frac{G_t (z_t^{a^i_t} )}{\sum\limits_{j=1}^N G_t (z_t^j )}$. 
The extended target is defined as
\begin{align*}
\pi^N_T &( dx_{0:T}^{1:N} , da_{0:T-1}^{1:N} , dn^{*}  ) \\
&= \frac{1}{\mathcal{Z}_T} m_0^{\otimes N} (dz^{1:N}_0)\\
&\times  \prod\limits_{t=1}^T \left\{  \left(  \frac{1}{N} \sum\limits_{n=1}^N G_{t-1} (z^n_{t-1} ) \right) \rho_{t-1} (z_{t-1}^{1:N} , da_{t-1}^{1:N}   )
\prod\limits_{n=1}^N q \left(dz_t^n \left| z_{t-1}^{a^n_{t-1}}  \right. \right) \right.\\ 
& \left. \times \frac{1}{N} G_T (z_T^{n^{*}})
 \right\}
\end{align*}
where
\[
\mathcal{Z}_T = m_{0}(dx_{0})\prod\limits _{i=0}^{n-1}dx_{i+1}G_{i}(x_{i})q(x_{i+1}\left|x_{i}\right) G_n (x_n )
\]

 It has been shown in \cite{Andrieu2010} that the marginal distribution of the random variable $Z^{n^{*}}_T$  sampled from $\pi^N_T $ is $\pi_T$. To construct a Gibbs sampler for this measure a conditional SMC Algorithm is used, i.e. an SMC Algorithm conditioned on a reference path $\bar{x}_{0},\ldots, \bar{x}_{n}$. The cSMC Algorithm takes as its input a fixed path, then creates a system of $N$ particles $X_{0:n}^{1:N}$ and ancestor indices 
$A_{0:n-1}^{1:N}$ 
To simplify the presentation of the Algorithm write $X^0_0 := \bar{x}_0, X^0_1 := \bar{x}_1, \ldots, X^0_n := \bar{x}_n$, i.e. particle zero is the fixed reference path. 

\begin{algorithm}\label{alg:csmcparticles}

\begin{itemize}
\item[ ]

\item Initialisation: for $i=1,\ldots, N$ set $X_{0}^{i}\sim m_{0}$ 
\item For $k=1,\ldots,n$
\begin{itemize}
\item Selection: For $i=1,\ldots,N$ select 
\[
A_{k-1}^{i}\sim\text{Multinomial}((w_{k-1}^{j})_{0\leq j\leq N})
\]
where 
\[
w_{k-1}^{i}=\frac{G_{k-1}(X_{k-1}^{i})}{G_{k-1} ( X^0_{k-1} ) + \sum\limits _{j=1}^{N}G_{k-1}(X_{k-1}^{j})}
\]

\item Propagation: For $i=1,\ldots,N$ 
\[
X_{k}^{i}\sim q\left( \cdot\left|X_{k-1}^{A_{k-1}^{i}} \right. \right)
\]
\end{itemize}
\end{itemize}
\end{algorithm}

The only modification of this Algorithm to standard SMC is in the selection step, since
the selection is over $X_{k}^{0},\ldots,X_{k}^{N}$ instead of $X_{k}^{1},\ldots,X_{k}^{N}$ which means that the fixed path can be selected for propagation. Algorithm \ref{alg:csmcparticles} can be used to define an algorithm which takes a conditional path $\bar{x}_{0:n}$ as its input and returns a new path.

\begin{algorithm}[cSMC]\label{alg:csmc}
\begin{itemize}
\item[ ] 

\item Create a system of particles $\left( X^{i}_{j} \right)_{1\leq i \leq N, 0 \leq j \leq n}$ and ancestor indices 

$\left( A^{i}_{j} \right)_{1\leq i \leq N, 0 \leq j \leq n-1}$ using Algorithm 1 conditioned on $\bar{x}_{0:n}$
\item Sample
\[
N^{*} \sim\text{Multinomial}((w_{n}^{j})_{0\leq j\leq N})
\]
where for $i=1,\ldots,N$
\[
w_{n}^{i}=\frac{G_{n}(X_{n}^{i})}{G_{n} ( X^0_{n} ) + \sum\limits _{j=1}^{N}G_{n}(X_{n}^{j})}
\]
\item Return the path $\left(X^{B_k^{N^*}}_k  \right)_{0\leq k \leq n} $ where the $B^i_k$ is defined by
\[
B^{i}_n = i
\]
and  for $k=n-1, \ldots, 0$
\[
B^{i}_k = A_{k}^{B^{i}_{k+1}}
\]
\end{itemize}
\end{algorithm}

Algorithm \ref{alg:csmc} can be regarded as a Markov kernel on $\mathcal{X}^{n+1}$ denoted by $P^N_n (\bar{x}_{0:n} , dx_{0:n})$ and it can be shown that this Markov kernel has $\pi_n$ as its invariant measure \cite{Chopin2012}. 


An important step in the analysis of cSMC will be the analysis of the perturbation of the selection step. For a more precise analysis of this define the Boltzmann-Gibbs transformation $\Psi_{k}$ for a probability measure $\mu$ on $\mathcal{X}^{k+1}$
\[
\Psi_{k} \mu (dx_{0:k}) = \frac{\mu(dx_{0:k}) G_k( x_k)}{\int \mu(dx_{0:k}) G_k( x_k ) } 
\]
Note that this is only defined if $\mu( G_k ) > 0$. 

For the empirical probability measure $\mu^N$ of the form $\mu^N (dx_{0:k}) = \frac{1}{N} \sum\limits_{i=1}^N \delta_{X^i_{0:k}} (d x_{0:k})$ this becomes
\[
\Psi_{k}(\mu^N )(dx_{0:k})=\sum\limits _{i=1}^{N}\frac{G_{k}(X^{i}_k)}{\sum\limits _{i=1}^{N}G_{k}(X^{i}_k)}\delta_{X^{i}_{0:k}} d(x_{0:k})
\]


To simplify the notation call $\delta_{\bar{x}_{0:k}}(dx_{0:k})$ simply $\delta_{k}(dx_{0:k})$. 
The modified Boltzmann-Gibbs transformation used in cSMC is
\[
\hat{\Psi}^N_{k}(\mu)=\Psi_{k}(N\mu + \delta_{k})
\]

Define 
\[
B^{i}_n = i
\]
and  for $k=n-1, \ldots, 0$
\[
B^{i}_k = A_{k}^{B^{i}_{k+1}}
\]

where the $A^{1:N}_{0:n}$ are as in Algorithm \ref{alg:csmcparticles}. Then let 
\begin{equation}\label{eq:empcsmc}
\hat{\eta}_{n}^{N}(dx_{0:n}) =\frac{1}{N} \sum\limits_{i=1}^N \delta_{(X^{B^i_{j}}_j)_{0\leq j \leq n}}
\end{equation}
be the empirical measure of the unconditioned
particle paths at time $k$ in Algorithm \ref{alg:csmcparticles}.

The empirical distribution of the selected particle of Algorithm \ref{alg:csmc} is given by $\hat{\Psi}^N_{n}(\hat{\eta}_{n}^{N})$ and one can write for the Markov kernel defined by Algorithm \ref{alg:csmc}
\begin{equation}\label{eq:markovkernel}
P_{n}^{N}(f_{n}) (\bar{x}_{0:n}) =\mathbb{E}_{\bar{x}_{0:n}} \left(\hat{\Psi}^N_{n}(\hat{\eta}^N_{n})(f_{n})\right)
\end{equation}



\section{Perturbation of the Selection step}\label{sec:perturb}

The first step is to quantify the perturbation induced by using the
modified Boltzmann-Gibbs transformation $\hat{\Psi}^N_{k}$ instead of $\Psi_{k}$. Remember that $\hat{\eta}^N_k$ is the empirical measure of the unconditioned particles created by Algorithm \ref{alg:csmcparticles} as defined in (\ref{eq:empcsmc}).

\begin{lemma}
Let $\varphi_{k}: \mathcal{X}^{k+1} \rightarrow \mathbb{R}$ be a function with $\left\Vert \varphi_{k}\right\Vert _{\infty}\leq1$. Then, 
\begin{align}
\mathbb{E}\vert\Psi_{k}(\hat{\eta}^N_{k})(\varphi_{k})-\hat{\Psi}^N_{k}(\hat{\eta}_{k}^{N})(\varphi_{k})\vert & \leq  \mathbb{E} ( Z_k ) \text{osc}(\varphi_{k})\label{eq:gibbs} 
\end{align}

where 
\begin{equation}\label{eq:Z}
Z_k = 1-\frac{N\hat{\eta}_{k}^{N}(G_{k})}{N\hat{\eta}_{k}^{N}(G_{k})+G_{k}(\bar{x}_{k})}
\end{equation}

\end{lemma}

\begin{proof}
Note that
\begin{align}
\Psi_{k}(\mu)-\hat{\Psi}^N_{k} (\mu) & =\left(1-\frac{\mu(G_{k})}{\mu(G_{k})+G_{k}(\bar{x}_{k})}\right)\Psi_{k}(\mu)-\frac{G_{k}(\bar{x}_{k})}{\mu(G_{k})+G_{k}(\bar{x}_{k})}\delta_{k}\nonumber \\
 & =\frac{G_{k}(\bar{x}_{k})}{\mu(G_{k})+G_{k}(\bar{x}_{k})}\left(\Psi_{k}(\mu)-\delta_{k}\right) \nonumber 
\end{align}

If we now substitute $\mu=\hat{\eta}_{k}^{N}$ this implies
\begin{align}
\mathbb{E}\vert\Psi_{k}(\hat{\eta}^N_{k})(\varphi_{k})-\hat{\Psi}^N_{k}(\hat{\eta}_{k}^{N})(\varphi_{k})\vert & \leq\mathbb{E}{\frac{G_{k}(\bar{x}_{k})}{N\hat{\eta}_{k}^{N}(G_{k})+G_{k}(\bar{x}_{k})}}\text{osc}(\varphi_{k})\nonumber \\
 & = \mathbb{E} ( Z_k ) \text{osc}(\varphi_{k})\label{eq:gibbs} 
\end{align}

\end{proof}


\section{Telescopic decomposition}
As explained in the previous section Algorithm \ref{alg:csmcparticles} can be regarded as a perturbed standard SMC Algorithm. 
Using ideas from \cite{delmoral2004} when analysing SMC methods we will give a telescopic decomposition for $P^N_n(f_n) - \pi_n (f_n ) $, where $P^N_n$ was defined in (\ref{eq:markovkernel}) as the Markov kernel for Algorithm \ref{alg:csmc}. This will be achieved by considering the individual local errors at each time step. Note that in the application of Gibbs sampling $P^N_n(f_n) - \pi_n (f_n ) $ is not of immediate interest, since accuracy of the Algorithm is achieved via stability properties of $(P^N_n )^m$. However, it will be shown that this error also gives a convergence result for repeated iterations of the cSMC kernel. 
Let $\mu$ be  a positive measure on $\mathcal{X}$. Then, define the perturbed forward kernel by 
\[
\hat{\Phi}^N_k \mu (dx_{0:k}) = \hat{\Psi}^N_{k-1} ( \mu ) (dx_{0:k-1}) dx_k q ( x_k \left| x_{k-1} \right )
\] 

\begin{theorem}\label{thm:telescopic}
For all $N\geq 1$, $n\geq 1$ and $f_n:\mathcal{X}^{n+1}\rightarrow \mathbb{R}$ with $\left\| f_n \right\|_{\infty} \leq 1$
\begin{equation}\label{eq:errorbound}
\left|P^N_n (f_n ) - \pi_n (f_n )\right| \leq \mathbb{E}(Z_n ) +  \frac{a_1}{N} \sum\limits_{k=1}^n  \mathbb{E}\left( \tau_{k,n}^2 \right) + \sum\limits_{k=1}^n \mathbb{E} ( \gamma_{k,n} ) 
\end{equation}
with
\[
\tau_{k,n} = \frac{\left\| Q_{k,n} (G_n ) \right\|_{\infty}}{\hat{\Phi}^N_k \hat{\eta}^N_{k-1}  \left( Q_{k,n} (G_n ) \right)}
\]
and
\[
\gamma_{k,n}=\frac{N^{-1}\delta_{k-1}Q_{k-1,n}(G_{n})}{(\hat{\eta}_{k-1}^{N}+N^{-1}\delta_{k-1})Q_{k-1,n}(G_{n})}
\]
\end{theorem}

\begin{proof}

Consider 

\begin{equation}
\left|P_{n}^{N}(f_{n})-\eta_{n}(f_{n})\right|
\leq\left|\mathbb{E}\left(\hat{\Psi}^N_{n}\hat{\eta}^N_{n}(f_{n})-\Psi_{n}\hat{\eta}_{n}^{N}(f_{n})\right)\right|+
\left|\mathbb{E}\left(\Psi_{n}\hat{\eta}_{n}^{N}(f_{n})\right)-\Psi_{n}\eta_{n}(f_{n})\right|\label{eq:first}
\end{equation}

The first part of this sum represents the perturbation of the Boltzmann-Gibbs measure at time $n$ and is bounded via (\ref{eq:gibbs})

For the second part of the sum in (\ref{eq:first}) we use a telescopic
decomposition. The idea is to split the  local error at each time step into one part which comes from the particle approximation error and can be analyzed using standard techniques from SMC theory and a second part capturing the distance between the perturbed and the unperturbed particle system.  Let 
\[
\mathcal{F}_{k}=\sigma\left\{ X_{1:k}^{1:N},A_{1:k-1}\right\} 
\]
Define
\[
V_{k}^{N}(dx_{0:k}) = \left(\hat{\eta}_{k}^{N}-\hat{\Phi}_{k}^N \hat{\eta}_{k-1}^{N} \right) (dx_{0:k})
\]

Let $\nu$ be a probability measure on $\mathcal{X}$. Define
\[
d_{\nu,k,n}(x_{0:k})=\frac{Q_{k,n}(G_{n})(x_k)}{\nu Q_{k,n}(G_{n})}P_{k,n}\left(f_{n}-  \frac{ \nu Q_{k,n} ( G_n f_n )}{\nu Q_{k,n} (G_n )}  \right)(x_{0:k})
\]

and note that $\nu( d_{\nu,k,n} ) = 0$. Write $d_{p,n} (f) (x) = d_{\eta_p, p, n} (f) (x) $ and use the conventions $\prod \emptyset = 1$ and $Q_{n,n} = Id$ and $\hat{\eta}^N_{-1} Q_{-1,n} = \eta_0 Q_{0,n} $.

Write 
\begin{align}
\Psi_n\ehb_{n}(\fb_{n}) &-\Psi_n\eta_{n}(\fb_{n}) \\
&=\sum\limits _{k=1}^{n}  \frac{(\hat{\eta}_{k}^{N})Q_{k,n}(G_{n} f_n)}{(\hat{\eta}_{k}^{N})Q_{k,n}(G_{n})} - \frac{(\hat{\Phi}^N_{k}\hat{\eta}_{k-1}^{N})Q_{k,n}(G_{n} f_n)}{(\hat{\Phi}^N_{k}\hat{\eta}_{k-1}^{N})Q_{k,n}(G_{n})}\\
&= \sum\limits _{k=1}^{n}\frac{(\hat{\Phi}^N_{k}\hat{\eta}_{k-1}^{N})Q_{k,n}(G_{n})}{\hat{\eta}_{k}^{N}Q_{k,n}(G_{n})}V_{k}^{N}(d_{(\hat{\Phi}^N_{k}\hat{\eta}_{k-1}^{N}),k,n})+R_{k}^{N} \label{eq:decomposition}
\end{align}
with the residual term
\begin{align}
R_{k}^{N} & =\left(\Psi_{n}\Phi_{k,n}\hat{\Phi}_{k}\hat{\eta}_{k-1}^{N}\right)(f_{n})-\left(\Psi_{n}\Phi_{k,n}\Phi_{k}\hat{\eta}^N_{k-1}\right)(f_{n})\nonumber \\
 & =\frac{(\hat{\eta}_{k-1}^{N}+N^{-1}\delta_{k-1})Q_{k-1,n}(G_{n}f_{n})}{(\hat{\eta}_{k-1}^{N}+N^{-1}\delta_{k-1})Q_{k-1,n}(G_{n})}-\frac{\hat{\eta}_{k-1}^{N}Q_{k-1,n}(G_{n}f_{n})}{\hat{\eta}_{k-1}^{N}Q_{k-1,n}(G_{n})}\nonumber \\
& =\left(\frac{\hat{\eta}_{k-1}^{N}Q_{k-1,n}(G_{n})}{(\hat{\eta}_{k-1}^{N}+N^{-1}\delta_{k-1})Q_{k-1,n}(G_{n})}-1\right)\frac{\hat{\eta}_{k-1}^{N}Q_{k-1,n}(G_{n}f_{n})}{\hat{\eta}_{k-1}^{N}Q_{k-1,n}(G_{n})}\nonumber \\
 & \qquad+\frac{N^{-1}\delta_{k-1}Q_{k-1,n}(G_{n})}{(\hat{\eta}_{k-1}^{N}+N^{-1}\delta_{k-1})Q_{k-1,n}(G_{n})}\frac{Q_{k-1,n}(G_{n}f_{n})(\bar{x}_{k-1})}{Q_{k-1,n}(G_{n})(\bar{x}_{k-1})}\nonumber \\ 
&= \frac{N^{-1}\delta_{k-1}Q_{k-1,n}(G_{n})}{(\hat{\eta}_{k-1}^{N}+N^{-1}\delta_{k-1})Q_{k-1,n}(G_{n})}
\left(     \frac{Q_{k-1,n}(G_{n}f_{n})(\bar{x}_{k-1})}{Q_{k-1,n}(G_{n}) (\bar{x}_{k-1})} - \frac{\hat{\eta}_{k-1}^{N}Q_{k-1,n}(G_{n}f_{n})}{\hat{\eta}_{k-1}^{N}Q_{k-1,n}(G_{n})}     \right) 
\end{align}
Hence, 
\begin{align}
\left | R^N_k \right|  &\leq\frac{N^{-1}\delta_{k-1}Q_{k-1,n}(G_{n})}{(\hat{\eta}_{k-1}^{N}+N^{-1}\delta_{k-1})Q_{k-1,n}(G_{n})}{\text{osc}}\left(\frac{Q_{k-1,n}(G_{n}f_{n})}{Q_{k-1,n}(G_{n})}\right) \nonumber \\
&\leq \gamma_{k,n} \label{eq:residualbound}
\end{align}

To analyse the first part of the sum (\ref{eq:decomposition}) note that in the filtration $\mathcal{F}_{k-1}$ the random variables $X_{k}^1, \ldots, X^N_k$ can be viewed as i.i.d. samples from $\hat{\Phi}^N_k \hat{\eta}_{k-1}^N$.  This implies that $\mathbb{E}\left( V_k^N \left| \mathcal{F}_{k-1} \right.  \right) = 0$ and we can apply standard results from \cite{delmoral2004} to control the particle approximation error. 

\begin{align*}
\mathbb{E} & \left( \frac{(\hat{\Phi}_{k}\hat{\eta}_{k-1}^{N})Q_{k,n}(G_{n})}{\hat{\eta}_{k}^{N}Q_{k,n}(G_{n})}V_{k}^{N}(d_{(\hat{\Phi}_{k}\hat{\eta}_{k-1}^{N}),k,n}) \right) \\
&=  \mathbb{E} \left( \left( \frac{(\hat{\Phi}_{k}\hat{\eta}_{k-1}^{N})Q_{k,n}(G_{n})}{\hat{\eta}_{k}^{N}Q_{k,n}(G_{n})} -1 \right)V_{k}^{N}(d_{(\hat{\Phi}_{k}\hat{\eta}_{k-1}^{N}),k,n}) \right)
\end{align*}

This allows to conclude that there exists a constant $a_1>0$ such that
\begin{equation}\label{eq:particle}
\left| \mathbb{E} \left(   \frac{(\hat{\Phi}_{k}\hat{\eta}_{k-1}^{N})Q_{k,n}(G_{n})}{\hat{\eta}_{k}^{N}Q_{k,n}(G_{n})}V_{k}^{N}(d_{(\hat{\Phi}_{k}\hat{\eta}_{k-1}^{N}),k,n}) \left| \mathcal{F}_{k-1} \right.    \right) \right| \leq \frac{1}{N} a_1 \frac{\left\| Q_{k,n} (G_n ) \right\|_{\infty}^2}{\hat{\Phi}^N_{k} \hat{\eta}^N_{k-1}  \left( Q_{k,n} (G_n ) \right)^2 }
\end{equation}

Using (\ref{eq:gibbs}), (\ref{eq:first}), (\ref{eq:decomposition}), (\ref{eq:particle}) and (\ref{eq:residualbound}) the  Theorem follows.

\end{proof}

\section{Explicit error bounds and uniform ergodicity}

Define for $0\leq k \leq n-1 $ 
\[
p_{k+1,n} (x_{k} ) = \int  \prod\limits_{j=k+1}^n dx_{j} q(x_j \left| x_{j-1}   \right ) G_j (x_j ) 
\]
In the framework of filtering for HMMs this corresponds to the probability density of the observations $y_{k+1:n}$ given $x_k$. 

In order to obtain explicit error bounds in (\ref{eq:errorbound}) we make the following assumptions
\begin{itemize}
\item[(A1)] There exists a constant $\underline{g}$ such that for all $0\leq k <n$ 
\begin{equation}\label{eq:condition}
\underline{g}^{-1}\geq \sup\limits_{x_{k}, x^\prime_k \in \mathcal{X}} \frac{ p_{k+1:n} (  x_k   )}{p_{k+1:n}(  x_k^\prime  )}\geq\underline{g}
\end{equation}
\item[(A2)] There exists a constant $a>0$ such that 
\[
\sup_{n\geq 0} \sup\limits_{x} G_n (x ) \leq a 
\] 
\end{itemize}

(A1) and (A2) implies that there exists a $c>0$ such that 
 for all $n>0$
\[
\int dx^\prime q(x^\prime \left| x \right. ) G_n (x^\prime) > c^{-1}
\]

To prove this, assume that the statement is false. Then there exists an increasing subsequence $( n_k )_{k\geq 0 }$ and a sequence $(x_{n_k} )_{k \geq 0}$ such that \[\int dx^\prime q(x^\prime \left| x_{n_{k}} \right. ) G_{n_k} (x^\prime) \rightarrow 0\] as $k \rightarrow \infty$. Take an arbitrary $\tilde{x} \in \mathcal{X}$.  Then by (A2) for all $k\geq 0$ it holds that $\int dx^\prime q(x^\prime \left| \tilde{x}  \right. ) G_{n_k} (x^\prime) \leq a$ and therefore 
$\frac{\int dx^\prime q(x^\prime \left| \tilde{x} \right. ) G_{n_k} (x^\prime)}{\int dx^\prime q(x^\prime \left| x_{n_{k}} \right. ) G_{n_k} (x^\prime) } \rightarrow \infty$ as $k \rightarrow \infty$, which contradicts (A1) for $k=n-1$.

These are standard assumptions in the analysis of SMC methods, see for example \cite{delmoral2004}, \cite{delmoral:hal-00932211}. They have also been applied to the analysis of maximum likelihood methods, see \cite{Douc2004} and \cite{douc2012}. In the context of filtering for HMMs one can show (A1) using 
lemma 2 in \cite{Douc2004} or proposition 8 in \cite{douc2012}.

\begin{theorem}\label{thm:totalerrorexplicit}
Under (A1),(A2), it holds that for all $N\geq 1$, $n\geq 1$ and $f_n:\mathcal{X}^{n+1}\rightarrow \mathbb{R}$ with $\left\| f_n \right\|_{\infty} \leq 1$

\[
\left|P_{n}^{N}(f_n )-\pi_{n}(f_n )\right|\leq 
\frac{2 (n+1) a c\underline{g}^{-1} - 1}{N-1+2ac\underline{g}^{-1}} + \frac{n (a c \underline{g}^{-1})^2 }{N} 
\]
\end{theorem}

 \begin{proof}
Note that under assumption \ref{eq:condition} it  holds that
\begin{align*}
\frac{p_{k,n} ( x_{k-1}   )}{\sup\limits_{x_k} p_{k+1:n} ( x_k   )} &= \int dx_{k} q(x_k \left| x_{k-1} \right ) G(x_k ) \frac{p_{k+1:n} ( x_k   )}{{\sup\limits_{x_k} p_{k+1:n} ( x_k   )}} \\
&> c^{-1} \underline{g}
\end{align*}

Hence,
 
\begin{align}
\tau_{k,n} &=  
\frac{\left\| Q_{k,n} (G_n ) \right\|_{\infty}}{\hat{\Phi}^N_k \hat{\eta}^N_{k-1}  \left( Q_{k,n} (G_n ) \right)}  \nonumber \\
&\leq  \dfrac{a \sup\limits_{x_k}  p _{k+1:n} ( x_k  ) }{ \hat{\Psi}^N_{k-1} \hat{\eta}^N_{k-1} (dx_{k-1}) p_{k:n} ( x_{k-1}   )  } \nonumber \\
&\leq a c \underline{g}^{-1} \label{eq:ineqtau}
\end{align}

In order to show that $\frac{N^{-1}\delta_{k-1}Q_{k-1,n}(G_{n})}{(\hat{\eta}_{k-1}^{N}+N^{-1}\delta_{k-1})Q_{k-1,n}(G_{n})}$ is of order $\frac{1}{N}$
we study 1 minus this quantity. First, note that $X_k^1, \ldots, X^N_k$ are i.i.d. under $\mathcal{F}_{k-1}$. 
Using Jensen's inequality it follows that for $1\leq i \leq N$

\[
\begin{alignedat}{1}\mathbb{E} & \left\{ \frac{Q_{k,n}(G_{n})(X_{k}^{i})}{(N\hat{\eta}_{k}^{N}+\delta_{k})Q_{k,n}(G_{n})}\vert\mathcal{F}_{k-1},X_{k}^{i}\right\} \\  
&\geq\frac{Q_{k,n}(G_{n})(X_{k}^{i})}{Q_{k,n}(G_{n})(X_{k}^{i})+\sum_{j=1,j\neq i}^{N}\mathbb{E}\left\{ Q_{k,n}(G_{n})(X_{k}^{j})\vert\mathcal{F}_{k-1},X_{k}^{i}\right\} + Q_{k,n}(G_{n})( \bar{x}_k )}\\
\\
 & =\frac{Q_{k,n}(G_{n})(X_{k}^{i})}{Q_{k,n}(G_{n})(X_{k}^{i})+(N-1)\hat{\Phi}_{k}(\hat{\eta}_{k-1}^{N})Q_{k,n}(G_{n})+Q_{k,n}(G_{n})(\bar{x}_k)}\\
 & \geq\frac{Q_{k,n}(G_{n})(X_{k}^{i})}{(N-1)\hat{\Phi}_{k}(\hat{\eta}_{k-1}^{N})Q_{k,n}(G_{n})+2\sup_{x_{k}\in \mathcal{X}}Q_{k,n}(G_{n})(x_{k})}\\
\\
\end{alignedat}
\]

Thus 
\begin{align*}
\mathbb{E} & \left\{ \frac{N\hat{\eta}_{k}^{N}Q_{k,n}(G_{n})(X_{k}^{i})}{(N\hat{\eta}_{k}^{N}+\delta_{k})Q_{k,n}(G_{n})} \left| \mathcal{F}_{k-1} \right. \right\} \\  & \geq
\frac{N}{N-1}\mathbb{E}\left\{ \frac{\hat{\Phi}^N_k  \hat{\eta}_{k-1}^{N} Q_{k,n}(G_{n})}{\hat{\Phi}^N_k  \hat{\eta}_{k-1}^{N} Q_{k,n}(G_{n})+\frac{2}{N-1}\sup_{x_{k}} Q_{k:n} ( G_n ) (  x_{k}  )}\right\} 
\end{align*}


Dividing the numerator and denominator of the previous fraction by $ \hat{\Phi}^N_k  \hat{\eta}_{k-1}^{N} Q_{k,n}(G_{n}) $ leads to
 
\[
\mathbb{E}\left( \frac{N\hat{\eta}_{k}^{N}Q_{k,n}(G_{n})(X_{k}^{i})}{(N\hat{\eta}_{k}^{N}+\delta_{k})Q_{k,n}(G_{n})} \right) 
\geq \frac{N}{N-1} \frac{1}{1+\frac{2}{N-1} \tau_{k,n}}
\]



Using (\ref{eq:ineqtau}) to bound $\tau_{k,n}$ it follows that
\begin{align*}
\frac{N^{-1}\delta_{k-1}Q_{k-1,n}(G_{n})}{(\hat{\eta}_{k-1}^{N}+N^{-1}\delta_{k-1})Q_{k-1,n}(G_{n})} &\leq 1 - \frac{N}{N-1} \frac{1}{1+\frac{2}{N-1} a c \underline{g}^{-1}} \\
&= \frac{2ac\underline{g}^{-1} - 1}{N-1+2a c \underline{g}^{-1}}  
\end{align*}

By the definition of $Z_n$ this also directly implies $\mathbb{E} (Z_n ) \leq \frac{2ac\underline{g}^{-1} - 1}{N-1+2a c \underline{g}^{-1}}  $. 

\end{proof}

In particular, this means that the error is of order $O \left( \frac{n}{N} \right)$. The Theorem can be used to obtain an ergodicity result for the cSMC kernel. Using the invariance property of the kernel it follows that

\begin{align*}
( (P^N_n)^m - \pi_n ) (f_n) &= (P^N_n)^{m-1}  ( P^N_n (f_n ) - \pi_n (f_n ) )\\
&= ( (P^N_n)^{m-1} - \pi_n ) ( P^N_n (f_n) - \pi_n (f_n) )\\
&\vdots \\
&= ( P^N_n (f_n ) - \pi_n (f_n) )^m 
\end{align*}

via induction. Thus,
\begin{corollary}\label{prop:stable}
Under (A1),(A2) it holds that for all $N\geq 1$, $n\geq 1$ and $f_n:\mathcal{X}^{n+1}\rightarrow \mathbb{R}$ with $\left\| f_n \right\|_{\infty} \leq 1$
\[
\left| ( (P^N_n)^m - \pi_n ) (f_n) \right| \leq \left( \frac{2 (n+1) a c \underline{g}^{-1} - 1}{N-1+2a c \underline{g}^{-1}} + \frac{n (a c \underline{g}^{-2})^2 }{N} \right)^m
\]
And by choosing $N_n=C (n+1)+1$ this error can be stabilized via
\[
\left| ( (P^{N_n}_n)^m - \pi_n ) (f_n) \right| \leq \frac{ a c \underline{g}^{-1} (2+ a c \underline{g}^{-1}  )}{C}
\]

\end{corollary}

\section{Central limit theorem}
It is well known that SMC satisfies a Central Limit Theorem (CLT) under mild assumptions, see \cite{delmoral2004}. In the following a telescopic decomposition is used to obtain a similar result for cSMC. We will write $\overset{d}{\rightarrow}$ for convergence in distribution and  $\overset{p}{\rightarrow}$ for convergence in probability.  

\begin{theorem}\label{thm:clt}
Suppose that $n \geq 0$ and that for all $0 \leq k < n$ it holds that $\left\|G_k \right\|_\infty < \infty$. Then, for any measurable $f_n : \mathcal{X}^{n+1} \rightarrow \mathbb{R}$ with $\left\| f_n \right\|_\infty < \infty$ and any conditional path $\bar{x}_{0:n}$
\[
\sqrt{N} ( \hat{\eta}^N_n (f_n ) - \eta_n (f_n ) ) \overset{d}{\rightarrow}  \mathcal{N} (0, \sigma_n^2   (f_n ) )
\]
as $N\rightarrow \infty$, where
\[
\sigma_n^2 (f_n ) = \eta_n \left( ( f_n - \eta( f_n ) )^2 \right) + \sum_{k=0}^{n-1} \eta_k \left(    \left( \frac{Q_{k,n}(f_n)}{\eta_k Q_{k,n}(1) } - \frac{\eta_k (Q_{k,n} f_n )}{\eta_k Q_{k,n}(1) }     \right)^2  \right)
\]
\end{theorem}

\begin{remark} 
Comparing Theorem \ref{thm:clt} with Theorem 15.5.1 in \cite{delmoral:hal-00932211}  shows that the cSMC method in Algorithm \ref{alg:csmcparticles} exhibits the same asymptotic behavior as an SMC filtering or smoothing Algorithm for $N \rightarrow \infty$. 
Under assumptions (A1) and (A2) the asymptotic variance can be bounded linearly in $n$. One can also use results from \cite{Whiteley2013} to obtain a linear bound in $n$ under much more general conditions.  
\end{remark}

\begin{proof}

Write 
\[\bar{G}_{k} (x )= \frac{G_k (x)}{\eta_k (G_k )}\] and define 
\[\bar{Q}_{k,n} (f_n )  (x_{0:k} ) = \frac{Q_{k,n} ( f_n ) (x_k ) }{\eta_k Q_{k,n} (1 )}
\]

Note that 
\begin{equation}\label{eq:product}
 \hat{\eta}^N_{p-1}  ( \bar{G}_{p-1} ) \times  \Phi_p (\hat{\eta}^N_{p-1}) \bar{Q}_{p,n} (f_n ) = 
\hat{\eta}^N_{p-1}  \bar{Q}_{p-1,n} (f_n ) 
\end{equation}
This implies
\begin{align*}
 \Phi_p \hat{\eta}^N_{p-1} \bar{Q}_{p,n} (f_n ) - \hat{\eta}^N_{p-1} \bar{Q}_{p-1,n} (f_n ) = (1-  \hat{\eta}^N_{p-1} (\bar{G}_{p-1} ) )\Phi_p \hat{\eta}^N_{p-1} \bar{Q}_{p,n} (f_n ) 
\end{align*}



Write

\[
W^N_n ( f ) :=  \hat{\eta}^N_n (f ) -  \eta_n (f_n ) \nonumber = \sum\limits_{p=0}^n \hat{\eta}^N_p \bar{Q}_{p,n} (f_n ) - \hat{\eta}^N_{p-1} \bar{Q}_{p-1,n} (f_n ) \nonumber 
\]

Using that $\hat{\eta}^N_n (f_n ) -  \eta_n (f_n ) =  \hat{\eta}^N_n (f_n -\eta_n(f_n ) ) -  \eta_n (f_n  -\eta_n(f_n ) ) $ and 
\begin{align*}
\hat{\eta}^N_p d_{p,n} (f_n ) - \Phi_p \hat{\eta}^N_{p-1} d_{p,n} (f_n ) &= \hat{\eta}^N_p  d_{p,n} (f_n ) - \hat{\Phi}_p \hat{\eta}^N_{p-1}  d_{p,n} (f_n ) \\ 
&+ \hat{\Phi}_p \hat{\eta}^N_{p-1} d_{p,n} (f_n ) - \Phi_p \hat{\eta}^N_{p-1} d_{p,n} (f_n )
\end{align*}
 this yields
\begin{equation}\label{eq:WNfirst}
W^N_n ( f_n ) = \sum\limits_{p=0}^n \hat{\eta}^N_p d_{p,n} (f_n ) - \hat{\Phi}^N_p \hat{\eta}^N_{p-1} d_{p,n} (f_n ) + L^N_n ( f_n ) + S^N_n (f_n ) 
\end{equation}

with the residual terms
\begin{align*}
L^N_n = 
\sum\limits_{p=1}^n \hat{\Phi}^N_p \hat{\eta}^N_{p-1} d_{p,n} (f_n )  - \Phi_p \hat{\eta}^N_{p-1} d_{p,n} (f_n )
\end{align*}
and 
\begin{align*}
S^N_n (f_n ) &=  \sum\limits_{p=1}^n (1- \hat{\eta}^N_{p-1}   (\bar{G}_{p-1} ) ) \times \Phi_p (\hat{\eta}^N_{p-1}) d_{p,n} (f_n ) \\
&=  \sum\limits_{p=1}^n (\eta_{p-1} (\bar{G}_{p-1}  ) - \hat{\eta}^N_{p-1}   (\bar{G}_{p-1} ) ) \times 
\left( \Phi_p \hat{\eta}^N_{p-1} - \Phi_p \eta_{p-1}  \right) d_{p,n} (f_n ) \\
&= \sum\limits_{p=1}^n (\eta_{p-1} (\bar{G}_{p-1}  ) - \hat{\eta}^N_{p-1}   (\bar{G}_{p-1} ) ) \times 
\frac{1}{\hat{\eta}^N_{p-1} (\bar{G}_{p-1})} \left( \hat{\eta}^N_{p-1} - \eta_{p-1} \right) d_{p-1,n} (f_n ) 
\end{align*}

where the second step has used that $\eta_{p-1} (\bar{G}_{p-1}  ) = 1$ and $( \Phi_p \eta_{p-1} )  d_{p,n} (f_n )=0$.







We will use these decompositions to prove the Theorem via an induction argument on $n$. First, note that by Theorem 9.3.1 in \cite{delmoral2004} it follows that for any bounded measurable $f_0 : \mathcal{X} \rightarrow \mathbb{R}$ it holds  $ \sqrt{N} \left( \hat{\eta}^N_0 ( f_0 ) - \eta_0 (f_0 ) \right) \overset{d}{\rightarrow} \mathcal{N} (0, \eta_0 ( ( \varphi - \eta_0 (f_0 ) )^2 )) $ since $\hat{\eta}^N_0$ is the empirical measure of $N$ i.i.d. samples from the initial distribution $\eta_0 = m_0$. Hence, the Theorem is proved for $n=0$. 

Next, assume that the Theorem is true for all $0\leq k <n$. This assumption will be called the inductive assumption. The target is to prove that this assumption implies that the theorem is also true for $k=n$. To do this, we consider the different parts of the sum in \ref{eq:WNfirst}.
 
Note that $\hat{\eta}_k^N$ is an empirical measure consisting of a sum of i.i.d. samples from $\hat{\Phi}_{k}^N \hat{\eta}_{k-1}^N$. 
Hence, one can proceed in a similar way as in  Theorem 9.3.1 in \cite{delmoral2004} (see also section 15.5 in \cite{delmoral:hal-00932211}   ). 
Let $\varphi_k : \mathcal{X}^{k+1} \rightarrow \mathbb{R}$ be a bounded measurable function.  
The idea is to Theorem 1 on page 171 in \cite{Pollard1984}, which is a central limit theorem for triangular arrays. To verify the conditions for this Theorem let $X^1_k, \ldots, X^N_k$ with $0\leq k \leq n$ be as in Algorithm \ref{alg:csmcparticles}. Define  $U^{i,N}_{k,n} = \frac{1}{\sqrt{N}} \left( d_{k,n} (f_n) (X^i_{0:k}) - \hat{\Phi}^N_k \hat{\eta}^N_{k-1} ( d_{k,n} ( f_n ) ) \right)$ for $1\leq i \leq N$.
Note that $\mathbb{E} ( U^{i,N}_{k,n} \left| \mathcal{F}_{k-1} \right. ) = 0$ and that
\[
\sqrt{N} \left( \sum\limits_{k=0}^n \hat{\eta}^N_k d_{k,n} (f_n ) - \hat{\Phi}^N_k \hat{\eta}^N_{k-1} d_{k,n} (f_n )\right) = \sum\limits_{k=0}^n \sum\limits_{i=1}^N     U^{i,N}_{k,n} 
\] 

Also due to $\left\| G_{k-1} \right\|_{\infty} < \infty $ and $\left\| \varphi_k \right\|_{\infty} < \infty$ the inductive assumption together with Slutsky's theorem 
\[
\hat{\eta}_{k-1}^N  ( G_{k-1} ) \overset{p}{\rightarrow} \eta_{k-1} ( G_{k-1} )
\]
as $N \rightarrow \infty$ and 
\[
\hat{\eta}_{k-1}^N  ( Q_{k-1,k} \varphi_k ) \overset{p}{\rightarrow} \eta_{k-1} ( Q_{k-1,k} \varphi_k )
\]
as $N \rightarrow \infty$. Using the above results together with Slutsky's theorem leads to 
\[
( \hat{\eta}_{k-1}^N + N^{-1} \delta_{k-1}  )  Q_{k-1,k} ( \varphi_k^2) \overset{p}{\rightarrow} \eta_{k-1} ( Q_{k-1,k} \varphi_k^2 )
\]
and 
\[
( \hat{\eta}_{k-1}^N + N^{-1} \delta_{k-1}  ) (G_{k-1})   \overset{p}{\rightarrow} \eta_{k-1} ( G_{k-1} ) 
\]

Hence, by the continuous mapping theorem
\begin{align*}
\hat{\Phi}_k^N \hat{\eta}^N_{k-1} ( \varphi_k^2) &= \frac{( \hat{\eta}_{k-1}^N + N^{-1} \delta_{k-1}  )  Q_{k-1,k} ( \varphi_k^2)  }{     ( \hat{\eta}_{k-1}^N + N^{-1} \delta_{k-1}  ) (G_{k-1})   } \\ 
&\overset{p}{\rightarrow} \eta_{k} ( \varphi_k^2 ) 
\end{align*}
and 
\[
( \hat{\Phi}_k^N \hat{\eta}^N_{k-1} ( \varphi_k ) )^2 \overset{p}{\rightarrow} ( \eta_{k-1} ( Q_{k-1,k} \varphi_k ) )^2 
\]

Moreover, note that for a fixed $n$ it holds that 
\[C_n := \sup\limits_{0\leq p \leq n } \sup\limits_{x_{0:p} \in \mathcal{X}^{p+1}} d_{p,n} (f_n ) (x_{0:p}) < \infty\] 
since $f_n$ is bounded and measurable. Therefore, for all $0\leq p \leq n$ $d_{p,n} (f_n )$ is a bounded measurable function. Hence, we can set $\varphi_k = d_{k,n} (f_n )$ and obtain from the above convergence results
\begin{align*}
\sum\limits_{k=0}^n \sum\limits_{i=1}^N  \mathbb{E} (  ( U^{i,N}_{k,n} )^2 \left| \mathcal{F}_{k-1} \right. ) &= \sum\limits_{k=0}^n \hat{\Phi}_k^N \hat{\eta}^N_{k-1} ( d_{k,n}^2 ) -  ( \hat{\Phi}_k^N \hat{\eta}^N_{k-1} ( d_{k,n}) )^2 \\
&\rightarrow  \sum\limits_{k=0}^n \eta_{k} (d_{k,n}^2 ) -  (  \eta_{k} ( d_{k,n}) )^2
\end{align*}

Note that  $\left|  U^{i,N}_{k,n}  \right|  \leq \frac{2 C_n }{\sqrt{N} }$ for all $1\leq i\leq N$, $0\leq k \leq n$. This implies that the Lindeberg condition is satisfied, i.e. for all $\epsilon >0 $
\[
\lim\limits_{N\rightarrow \infty} \mathbb{E} (  \sum\limits_{k=0}^n  ( U^{i,N}_k )^2  1_{\left| U^{i,N}_k \right| > \epsilon } \left| \mathcal{F}_{k-1} \right. )  = 0 
\]
and therefore Theorem 1 on page 171 in \cite{Pollard1984} can be applied to obtain 
\begin{align}
\sqrt{N} & \left( \sum\limits_{p=0}^n \hat{\eta}^N_p d_{p,n} (f_n ) - \hat{\Phi}^N_p \hat{\eta}^N_{p-1} d_{p,n} (f_n )\right) \nonumber \\ 
& \overset{d}{\rightarrow} 
\mathcal{N} \left( 0,   \sum\limits_{k=0}^n \eta_k \left(    \left( \frac{Q_{k,n}(f_n)}{\eta_k Q_{k,n}(1) } - \frac{\eta_k (Q_{k,n} f_n )}{\eta_k Q_{k,n}(1) }     \right)^2  \right) \right) \label{eq:convergence}
\end{align}

 It remains to show that $\sqrt{N} S^N_n (f_n ) \overset{d}{\rightarrow} 0$ and $\sqrt{N}  L^N_n (f_n ) \overset{d}{\rightarrow} 0$ as $N\rightarrow \infty$ .  

Now consider the terms in $\sqrt{N} S^N_n(f_n )$. Note that by the inductive assumption and the boundedness of the $G_k$ for $0\leq k< n$  it holds
$\sqrt{N} (\eta_{k-1} (G_{k-1}  ) - \hat{\eta}^N_{k-1}   ( G_{k-1} ) ) \overset{d}{\rightarrow} \mathcal{N} (0, \sigma_{k-1} (  G_{k-1} ) )$ and therefore $\hat{\eta}^N_{k-1} (\bar{G}_{k-1}) \overset{P}{\rightarrow} \eta_{k-1} (\bar{G}_{k-1})$ and  $ (\hat{\eta}^N_{k-1} - \eta_{k-1} ) d_{k-1,n} (f_n ) \overset{P}{\rightarrow} 0$. By Slutsky's Theorem it follows that \[\frac{1}{\hat{\eta}^N_{k-1} (\bar{G}_{k-1})} \left( \hat{\eta}^N_{k-1} - \eta_{k-1} \right) d_{k-1,n} (f_n ) \overset{P}{\rightarrow} 0\] 
and again by Slutsky's Theorem 
\[
\sqrt{N } (\eta_{k-1} (\bar{G}_{k-1}  ) - \hat{\eta}^N_{k-1}   (\bar{G}_{k-1} ) ) \times \frac{1}{\hat{\eta}^N_{k-1} (\bar{G}_{k-1})} \left( \hat{\eta}^N_{k-1} - \eta_{k-1} \right) d_{k-1,n} (f_n ) \overset{d}{\rightarrow} 0 \] 
which implies $\sqrt{N}S^N_n (f_n ) \overset{d}{\rightarrow} 0$

To consider $\sqrt{N} L^N_n (f_n)  $ suppose $\varphi_k: \mathcal{X}^{k+1} \rightarrow \mathbb{R}$ is a bounded measurable function. Then,
\begin{align*}
\sqrt{N} & \left( \hat{\Phi}_k^N \hat{\eta}^N_{k-1} \varphi_k  - \Phi_k \hat{\eta}^N_{k-1} \varphi_k \right) \\
&=  N^{-1/2} \frac{    \hat{\eta}_{k-1}^N (G_{k-1}) \times Q_{k-1,k} (\varphi_k ) (\bar{x}_{0:k-1} ) - \hat{\eta}_{k-1}^N (Q_{k-1, k} (\varphi_k )) \times G_{k-1}  (\bar{x}_{k-1} )        }{(\hat{\eta}^N_{k-1} +  N^{-1} \delta_{k-1} ) ( G_{k-1} ) \times \hat{\eta}^N_{k-1} (G_{k-1} ) }\\
&\overset{p}{\rightarrow} 0
\end{align*}
as $N \rightarrow 0$. The convergence follows from applying the inductive assumption together with the boundedness of $G_{k-1}$ and $\varphi_k$ to obtain $\hat{\eta}_{k-1}^N (G_{k-1}) \overset{p}{\rightarrow} \eta_{k-1} (G_{k-1}) $ and $\hat{\eta}_{k-1}^N (Q_{k-1, k} (\varphi_k )) \rightarrow \eta_{k-1} (Q_{k-1, k} (\varphi_k )) $ and then applying the continuous mapping theorem to the fraction in the above equation. It follows that $ \hat{\Phi}_k \hat{\eta}^N_{k-1} d_{k,n} (f_n) - \Phi_k \hat{\eta}^N_{k-1} d_{k,n} (f_n) \overset{p}{\rightarrow} 0$ and therefore $\sqrt{N} L^N_n (f_n ) \overset{p}{\rightarrow} 0$.

We have derived convergence results for all three parts of (\ref{eq:WNfirst} ). By Slutsky's theorem they can be added to obtain
\[
W^N_n ( f_n ) \overset{d}{\rightarrow} 
\mathcal{N} \left( 0,   \sum\limits_{k=0}^n \eta_k \left(    \left( \frac{Q_{k,n}(f_n)}{\eta_k Q_{k,n}(1) } - \frac{\eta_k (Q_{k,n} f_n )}{\eta_k Q_{k,n}(1) }     \right)^2  \right) \right)
\]

This completes the induction step.


\end{proof}

\section{Discussion}
This paper shows that considering cSMC as a conventional SMC Algorithm with a perturbed Boltzmann-Gibbs measure can lead to interesting new insights into the method which presumably cannot be obtained using the minorization or coupling techniques of  \cite{Andrieu2013}, \cite{Lindsten2014} or \cite{Chopin2012}. In particular, it is possible to derive a central limit theorem and to prove that SMC and cSMC converge to the same asymptotic distribution. The observation that asymptotically the distribution of the particles is close to the target distribution can be used to motivate more sophisticated pMCMC algorithms which use the particle system simulated by cSMC to adapt the proposal distributions. 

In the following we compare our ergodicity and stability results to the ones which are already available in the literature. First note that our assumptions (A1), (A2) are very similar to the ones used by \cite{Andrieu2013} where they derive a bound for the ergodicity coefficient of cSMC under the assumption $\sup\limits_{x \in \mathcal{X}} Q_{k,n}(1)(x) / \eta_k Q_{k,n}(1) \leq \alpha$. \cite{Lindsten2014} use the slightly stronger assumption $1\leq \sup_{x\in\mathcal{X}}G_k (x) / \inf_{y\in\mathcal{X}} G_k (y) \leq \delta$. 

To compare our stability result in proposition \ref{prop:stable} with other results, note that \cite{Andrieu2013} prove that under their conditions the ergodicity coefficient is 

$\left( \frac{N-1}{N+ 2(\underline{g}-1 )}  \right)^n$. For $N-1 \geq C n$ this becomes $\exp \left(  \frac{2\alpha -1}{C}  \right)$ asymptotically as $n$ goes to infinity. Hence, we can conclude that both their and our proof technique show that cSMC can be stabilized by increasing the number of particles $N$ linearly with $n$. 

To obtain a more detailed analysis of the cSMC particle system it would be possible to apply the telescopic sum decomposition together with our definition of the perturbed Boltzmann-Gibbs measure to derive $\mathbb{L}^p$-error bounds, concentration and propagation of chaos inequalities, similar to the results in \cite{delmoral2004} . 

An interesting extension of this paper would be to consider cSMC with backward resampling. Numerical studies show that cSMC for multinomial and residual resampling the use of backward resampling techniques leads to a superior performance of the cSMC method \cite{Chopin2012}. One could combine the approach from this paper with \cite{jasra2013behaviour} to study these techniques.

\bibliographystyle{abbrvnat}	

\end{document}